\documentclass[11pt]{article}
\usepackage[dvips]{graphicx}\usepackage{latexsym}
\usepackage[]{amsmath}\usepackage{amsfonts,amssymb}
\usepackage{color}

\newtheorem{theorem}{Theorem}\newtheorem{lemma}{Lemma}
\newtheorem{definition}{Definition}\newtheorem{proposition}{Proposition}
\newtheorem{corollary}{Corollary}\newtheorem{claim}{Claim}

\def\squarebox#1{\hbox to #1{\hfill\vbox to #1{\vfill}}}
\def\qed{\hspace*{\fill}        \vbox{\hrule\hbox{\vrule\squarebox{.667em}\vrule}\hrule}\smallskip}
\newenvironment{proof}{\begin{trivlist}
  \item[\hspace{\labelsep}{\em\noindent Proof.~}]  }{\qed\end{trivlist}}

 \newcommand{\bs}{\bigskip} 
 \newcommand{\n}{\noindent} 
 \newcommand{\hs}[1]{\hspace*{ #1 mm}} 

\newenvironment{proofof}[1]{\vspace*{5mm} \par \noindent
         {\bf Proof of #1.\hs{2}}}{\hfill$\Box$ \vspace*{3mm}}
\newcommand{\ignore}[1]{}
\ignore{\documentclass[11pt]{article}%\documentclass[titlepage,11pt,twocolumn]{article}
\usepackage{amstext,amsgen,latexsym,amsmath}\usepackage{amstext,amssymb,amsfonts,latexsym}\usepackage{theorem} \usepackage{pifont}\usepackage[dvips]{graphics,epsfig}\setlength{\evensidemargin}{0cm}\setlength{\oddsidemargin}{0cm}\setlength{\topmargin}{0cm}\setlength{\textheight}{9in}\setlength{\textwidth}{6.5in}\setlength{\leftmargin}{0mm}\setlength{\parsep}{1mm}\setlength{\itemsep}{1mm}\setlength{\itemindent}{1mm}\setlength{\topsep}{1mm}\setlength{\labelsep}{3mm}\setlength{\parskip}{0mm}\setlength{\listparindent}{0mm}\setlength{\headsep}{0cm}\setlength{\headheight}{0cm}\setlength{\marginparwidth}{0cm}\setlength{\topskip}{-1.0cm}
 \newcommand{\bs}{\bigskip}  \newcommand{\n}{\noindent}  \newcommand{\hs}[1]{\hspace*{ #1 mm}} 
 \def\bbox{\vrule height6pt width6pt depth1pt}\theoremstyle{plain}\theoremheaderfont{\bfseries}\setlength{\theorempreskipamount}{3mm}\setlength{\theorempostskipamount}{3mm}
 \newtheorem{theorem}{Theorem}[section] \newtheorem{lemma}[theorem]{Lemma}       \newenvironment{proof}{\par \noindent            {\bf Proof. \hs{2}}}{\hfill$\Box$ \vspace*{3mm}} 
\setcounter{page}{1}     
}

\begin{document}
%%%%%%%%%%%%%%%%%%
\pagestyle{plain}
\begin{center}
{\Large {\bf Quantum Counterfeit Coin Problems}}
%\footnote{} 
\bs\\

{\sc Kazuo Iwama}$^1$\footnote{Supported in part by KAKENHI (19200001,22240001)} \hspace{5mm}
{\sc Harumichi Nishimura}$^2$\footnote{Supported in part by KAKENHI (21244007,22700014)} \hspace{5mm} 
{\sc Rudy Raymond}$^3$ \hspace{5mm} 
{\sc Junichi Teruyama}$^1$ 

\

{\small
$^1${School of Informatics, Kyoto University, Japan}; {\tt $\{{\tt iwama,teruyama}\}$@kuis.kyoto-u.ac.jp} 

$^2${School of Science, Osaka Prefecture University, Japan}; {\tt hnishimura@mi.s.osakafu-u.ac.jp}

$^3${IBM Research -- Tokyo, Japan}; {\tt raymond@jp.ibm.com} 
}

\end{center}
\bs

\n{\bf Abstract.}\hs{1} 
The counterfeit coin problem requires us to find all false coins from
a given bunch of coins using a balance scale.  We assume that the
balance scale gives us only ``balanced'' or ``tilted'' information and 
that we know the number $k$ of false coins in advance.  The balance
scale can be modeled by a certain type of oracle and its query
complexity is a measure for the cost of weighing algorithms (the
number of weighings).  In this paper, we study the quantum query
complexity for this problem.  Let $Q(k,N)$ be the quantum query
complexity of finding all $k$ false coins from the $N$ given coins.
We show that for any $k$ and $N$ such that~$k < N/2$, $Q(k,N)=O(k^{1/4})$, 
contrasting with the classical query complexity, $\Omega(k\log(N/k))$, that depends on 
$N$. So our quantum algorithm achieves a {\it quartic} speed-up for this problem. 
We do not have a matching lower bound, but we show some evidence that the upper bound 
is tight: any algorithm, including our algorithm, that satisfies certain properties 
needs $\Omega(k^{1/4})$ queries. 

\section{Introduction} 

{\em Exponential} speed-ups by quantum algorithms have been highly
celebrated, but their specific examples are not too many.  
In contrast, almost every unstructured search problem can be sped up 
simply by using amplitude amplification~\cite{Gro96,BBHT98,BHMT02}, 
providing a huge number of combinatorial problems 
for which quantum algorithms are {\it quadratically} faster than classical ones. 
Interestingly there are few examples in between. 
(For instance, \cite{DS08} provides a cubic speed-up while 
their classical lower bound is not known.) The reason is probably that the
amplitude amplification is too general to combine with other
methods appropriately. In fact we know few such cases including the
one by \cite{FSS07} where they improved a simple Grover search
algorithm for triangle finding by using clever combinatorial ideas
(but unfortunately still less than quadratically compared to the best
classical algorithm). This paper achieves a {\it quartic}
speed-up for a well-known combinatorial problem.

{\em The counterfeit coin problem} is a mathematical puzzle whose origin
dates back to 1945; in the American Mathematical Monthly,~52,~p.~46,
E.~Schell posed the following question which is probably one of the
oldest questions about the complexity of algorithms: ``You have eight
similar coins and a beam balance.  At most one coin is counterfeit and
hence underweight. How can you detect whether there is an underweight
coin, and if so, which one, using the balance only twice?''  The puzzle
immediately fascinated many people and since then there have been
several different versions and extensions in the literature
(see e.g.,~\cite{Man77,GN95,HH95,LZN05}).

This paper considers the quantum version of this problem, which, a bit
surprisingly, has not appeared in the literature. To make our model
simple, we assume that we cannot obtain information on which side is heavier when
the scale is tilted. So, the balance scale gives us only binary information,
{\em balanced} (i.e., two sets of coins on the two pans are equal in
weight) or {\em tilted} (different in weight). Our goal is to detect the false coin with a
minimum number of weighings.  The problem is naturally extended to the
case that there are two or more ($=k$ that is known in advance) false
coins with equal weight.  For the simplest case that $k=1$, the
following easy (classic) algorithm exists: We put (approximately)
$N/4$ coins on both pans. If the scale is tilted, then we know the false coin
is in those $N/2$ coins and if it is balanced, then the false one should
be in the remaining $N/2$ ones. 
Also, it is easy to see that two weighings are enough for $N=4$. 
Thus $\lceil \log N \rceil$ weighings
are enough for $k=1$ and this is also an information theoretic lower bound. 
(The original version of the problem assumes ternary outputs from the balance,
left-heavy, right-heavy and balanced, and that the false coin is always underweight.
As one can see easily, however, the same idea allows us to obtain 
the tight upper bound of $\lceil \log_3 N \rceil$.)

Our model of a balance scale is a so-called {\em oracle}. 
{\em A balance oracle} or simply {\em a B-oracle} is an $N$-bit register, which
includes (originally unknown) $N$ bits, $x_1x_2\cdots x_N \in
\{0,1\}^N$.  In order to retrieve these values, we can make {\em a
query} with {\em a query string} $q_1q_2 \cdots q_N \in \{0, 1,
-1\}^N$ including the same number ($=l$) of $1$'s and $-1$'s.  
Then the oracle returns a one-bit answer $\chi$ defined as: \vspace{-2mm}

$$
\chi=0 \mbox{ if } x_1q_1 + \cdots + x_Nq_N =0 \mbox
{ and }\chi=1 \mbox{ otherwise.}
$$\vspace{-4mm}

\noindent
Consider $x_1, \cdots, x_N$ as $N$ coins where $0$ means a
fair coin and $1$ a false one.  Then, $q_i=1$ means we place coin $x_i$ on
the left pan and $q_i=-1$ on the right pan.  Since we must have the same
number of $1$'s and $-1$'s, the answer $\chi$ correctly simulates the
balance scale, i.e., $\chi=0$ means it is balanced and $\chi=1$ tilted. 
The number of weighings needed to retrieve $x_1$ through $x_N$ (or to
identify all the false coins) is called {\em query complexity}.

The main purpose of this paper is to obtain {\em quantum} query complexity 
for the counterfeit coin problems.  Observe that if we know in advance that an
even-cardinality set ${X}$ includes {\em at most one} false coin, then by
using the balance for any equal-size partition of ${X}$ we can get the
parity of ${X}$, i.e., the parity of the number (zero or one, now) of false coins in ${X}$. 
This means that for strings including at most one 1, the B-oracle is
equivalent to the so-called IP oracle \cite{BV97}. Therefore, 
by Bernstein-Vazirani algorithm \cite{BV97}, we need only one weighing to
detect the false coin. (Note that this observation was essentially done by Terhal and Smolin \cite{TS98}.) 
This already allows us to design the following 
quantum algorithm for general~$k$: Recall that we know $k$ in advance. 
So, if we sample $N/k$ coins at random, then they include exactly one false
coin with high probability and we can find it using the
B-oracle just once as mentioned above.  Thus, by using the standard
amplitude amplification \cite{BHMT02} (together with a bit careful consideration for
the answer-confirmation procedure), we need $O(k)$ weighings to find all
$k$ false coins. For a small $k$, this is already much better than $\Omega(k\log(N/k))$
that is an information theoretic lower bound for the classical case.

{\bf Our Contribution.} 
This paper shows that this complexity can be
furthermore improved quartically, namely, our new algorithm needs
$O(k^{1/4})$ weighings.  Note that the above idea, the one exploiting
Bernstein-Vazirani, already breaks down for $k=2$, since the scale
tilts even if the pans hold two (even) false coins if they both go to
a same pan.  Moreover, if $k$ grows, say as large as linear in $N$,
the balance will be tilted almost always for randomly selected equal
partitions.  Nevertheless, Bernstein-Vazirani is useful since it
essentially reduces our problem (identifying false coins) to the
problem of deciding the parity of the number of the false coins that
turns out to be an easier task for B-oracles.  By this we can get a
single quadratic speed-up and another quadratic speed-up by amplitude amplification. 

We conjecture that this bound is tight, but unfortunately, 
we cannot prove it at this moment. The main difficulty is that we have a
lot of freedom on ``the size of the pans'' (= the number of coins
placed on the two pans of the scale), which makes it hard to design a
single weight scheme of the adversary method~\cite{Amb02}. However, we do have
a proof claiming that we cannot do better unless we can remove the two
fundamental properties of our algorithm.  These properties are (i) the
big-pan property and (ii) the random-partition property. 
We have considered several possibilities for escaping from them, 
but not successful for even one of them. 

{\bf Related Work.} Query complexities have been studied almost always
for the standard {\em index oracle}, which accepts an index $i$ and returns
the value of $x_i$.  Other than this oracle, we know few ones
including the IP oracle \cite{BV97} mentioned before and the even more powerful
one that returns the number (not the parity) of 1's in the string \cite{TS98}.
Also, \cite{TS98} presented a single-query quantum algorithm 
for the binary search problem under the IP oracle, which is essentially 
based on the same idea as the $k=1$ case of our problem mentioned above.

The quantum adversary method, which is used for B-oracles in this
paper, was first introduced by Ambainis \cite{Amb02} for the
standard oracle.  Many variants have followed including weighted adversary methods
\cite{Amb06,Zha04}, spectral adversary method \cite{BSS03}, Kolmogorov
complexity method \cite{LM08}, all of which were shown to be
equivalent \cite{SS06}.  After H{\o}yer et al.~\cite{HLS07} introduced
a stronger quantum adversary method called the negative adversary method, 
Reichardt \cite{Rei09,Rei10} showed that this method is ``optimal'' for any Boolean function.

{\bf Models.} A {\em B-oracle} is a binary string $x=x_1\cdots x_N$ where $x_i=1$ (resp.~$=0$) 
means that the $i$-th coin is false (resp.~fair). For instance, the string $0001$ for $N=4$ means
that the fourth coin is a unique false coin. A query to the oracle is given as a string 
$q=q_1\cdots q_N\in \{0,1,-1\}^N$ that must be in the set 
$Q^{(B)}=\bigcup_{l=0}^{\lfloor N/2\rfloor}Q_l$ 
where $Q_l$ is the set of strings $q$ such that $q$ has exactly~$l$~$1$'s and $l$~${-1}$'s.  
Here, $1$ (or $-1$, resp.) in the $i$-th component means that we place the $i$-th coin on the left pan
(on the right pan, resp.) and $0$ means that the $i$-th coin is not placed on either pan. 
The answer from the oracle is represented by a binary value $\chi(x;q)$ 
where $\chi(x;q)=0$ means the scale is balanced, that is, $q_1x_1+\ldots+q_Nx_N=0$ 
and $\chi(x;q)=1$ means it is tilted, that is, $q_1x_1+\ldots+q_Nx_N\neq 0$. 
In quantum computation, the B-oracle is viewed as a unitary transformation~$O_{B,x}$.
Namely,~$O_{B,x}$ transforms $|q\rangle$ to $(-1)^{\chi(x;q)}|q\rangle$. 
Throughout this paper, we assume that~$k<N/2$ since our B-oracle model 
is unable to distinguish any $N$-bit string~$x$ from~$\bar{x}$ 
(the bit string obtained by flipping all bits of $x$).  

\section{Upper Bounds}\label{one-fourth}
%\subsection{Upper Bound Result}
Here is our main result in this paper:

\begin{theorem}\label{main_theorem}
The quantum query complexity for finding $k$ false coins among $N$ coins is $O(k^{1/4})$.
\end{theorem}

Notice that our algorithm is {\em exact}, i.e., its output must be correct with probability one to compare our result 
with the classical case (which has been often studied in the exact setting). Since 
we use exact amplitude amplification \cite{BHMT02} to make our algorithm exact, 
the assumption that $k$ is known is necessary. But it should be noted 
that our bounded-error algorithm described in this section works even 
for unknown $k$. Also, we note that our algorithm can be easily adapted 
so that it works when the output of the balance is ternary 
(while we assume it is binary for simplicity).

Before the proof, we first describe our basic approach, a simulation of the IP oracle by the B-oracle. 
Recall that the IP oracle (Inner Product oracle) \cite{BV97} transforms a prequery state
$|\widetilde{q}\rangle_{{\sf R}}$ to $(-1)^{\widetilde{q}\cdot x}|\widetilde{q}\rangle_{{\sf R}}$, 
where $\widetilde{q}\in\{0,1\}^N$ in register~${\sf R}$ is a query string 
and $x\in\{0,1\}^N$ is an oracle. 
Then the Bernstein-Vazirani algorithm (the Hadamard transform) retrieves
the string $x$ and we know the $k$ false coins in the case of our problem.
Observe that the IP oracle flips the phase of each state
if and only if~$\widetilde{q}\cdot x$ is odd, in other words,
if and only if a multiset $M(\widetilde{q},x):=\{x_i \mid \widetilde{q}_i = 1\}$
includes an odd number of $1$'s (or false coins in our case).
If $k=1$, then $M(\widetilde{q},x)$ includes at most one $1$.
Hence we can simply replace the IP oracle with the query sequence $\widetilde{q}$ 
by the B-oracle with a query sequence~$q$
such that an arbitrarily one half (the first one half, for instance) of the $1$'s in $\widetilde{q}$
are changed to  $-1$'s, meaning the one half of the coins in $M(\widetilde{q},x)$ go to the left pan
and the remaining one half to the right pan. (As shown in a moment, we can assume without loss of generality 
that $\widetilde{q}$ includes an even number of $1$'s.)

Now we consider the general ($k \geq 1$) case. If $M(\widetilde{q},x)$ includes odd $1$'s, 
then the scale is tilted for any such $q$ mentioned above; this is desirable for us.
If $M(\widetilde{q},x)$ includes even $1$'s, we wish the scale to be balanced.
In order for this to happen, however, we must divide the (unknown) false coins in $M(\widetilde{q},x)$
into the two pans evenly, for which there are no obvious ways other than using randomization.
Our idea is to introduce the second register, ${\sf R}'$, as follows:
On ${\sf R}'$, we prepare, with being entangled to each state $\widetilde{q}$ in ${\sf R}$,
a superposition of all possible states $q_1(\widetilde{q}), q_2(\widetilde{q}), \ldots, q_h(\widetilde{q})$,
obtained by flipping one half of $1$'s in $\widetilde{q}$ into ${-1}$'s.
By using this superposition as a query to the B-oracle,
we can achieve a success (being able to detect the scale is balanced) probability of~$1/\sqrt{m}$,
where $m$ is the number of false coins in $M(\widetilde{q},x)$.
In order to increase this probability, we can use copies of register ${\sf R}'$
or, more efficiently, quantum amplitude amplification~\cite{BHMT02}.

As suggested before, we begin with the restriction of the IP oracle without losing its power.
The {\em parity-restricted query} means that the Hamming weights of
all superposed queries $\widetilde{q}$, denoted by $wt(\widetilde{q})$, are even.

\begin{lemma}\label{parity-restricted}
Let $S_{<N/2}:=\{x\in\{0,1\}^N\mid wt(x)< N/2\}$. Then there is a quantum algorithm to identify
an oracle in $S_{<N/2}$ by a single parity-restricted query for the IP oracle.
\end{lemma}

\begin{proof}
For a given oracle $x \in S_{<N/2}$, define
\[
|\psi_x\rangle
=\frac{1}{\sqrt{2^{N-1}}}\sum_{\widetilde{q}\in Q_{even}}(-1)^{\widetilde{q}\cdot x}|\widetilde{q}\rangle.
\]
where $Q_{even}=\{\widetilde{q}\in\{0,1\}^N \mid wt(\widetilde{q})=0\ \mathrm{mod}\ 2\}$.
%by a single parity-restricted query to the IP oracle. %(where we omit the normalized coefficient).
Then the Hadamard transform of $|\psi_x\rangle$, $H|\psi_x\rangle$, can be rewritten as follows:
\begin{align*}
H|\psi_x\rangle
&= \frac{1}{\sqrt{2^{N-1}}}\sum_{\widetilde{q}\in Q_{even}} (-1)^{\widetilde{q}\cdot x} H|\widetilde{q}\rangle
\ =\ \frac{1}{2^{N-1}\sqrt{2}} \sum_{\widetilde{q}\in Q_{even}}\sum_{z \in \{0,1\}^N}(-1)^{\widetilde{q}\cdot(x \oplus z)}|z\rangle\\
%&= \frac{1}{2^{N-1}\sqrt{2}}\sum_{\widetilde{q}\in Q_{even}}(|x\rangle+(-1)^{wt(\widetilde{q})}|\bar{x}\rangle)
%+\frac{1}{2^{N-1}\sqrt{2}}\sum_{\widetilde{q}\in Q_{even}}\sum_{z\neq x,\bar{x}}
%(-1)^{\widetilde{q}\cdot(x \oplus z)}|z\rangle\\
&= \frac{1}{\sqrt{2}} \left(|x\rangle + |\bar{x}\rangle\right)
+\frac{1}{2^{N-1}\sqrt{2}}\sum_{\widetilde{q}\in Q_{even}}\sum_{z\neq x,\bar{x}}
(-1)^{\widetilde{q}\cdot(x \oplus z)}|z\rangle\\ 
&=\frac{1}{\sqrt{2}} \left(|x\rangle + |\bar{x}\rangle\right).
\end{align*}
Note that the last equality in the above equations holds; 
the second term must vanish because the first term already has a unit length. 
For any $x\neq y$, $H|\psi_x\rangle=(|x\rangle + |\bar{x}\rangle)/\sqrt{2}$ 
and $H|\psi_y\rangle=(|y\rangle + |\bar{y}\rangle)/\sqrt{2}$ are orthogonal 
since $x \neq \bar{y}$ by the restriction of their Hamming weights. 
This implies that $|\psi_x\rangle$ is orthogonal to~$|\psi_y\rangle$ for
any $x \neq y$, and hence there is a unitary transformation $W: |x\rangle\mapsto |\psi_x\rangle$. 
Thus we can design an algorithm similar to Bernstein-Vazirani \cite{BV97} 
just replacing the Hadamard transform by $W$. 
For a concrete (polynomial-time) construction of~$W,$ see Appendix \ref{make_w}. %\hfill$\Box$
\end{proof}

Now we give the proof of our main result.

\begin{proofof}{Theorem \ref{main_theorem}}%\subsubsection{Bounded Error Algorithm}\label{bounded-error}
For exposition, we first give a bounded-error algorithm ($Find^*(k)$) 
and then make it exact ($Find(k)$). 
In what follows, for a query string~$\widetilde{q}$, 
let $I(\widetilde{q})$ be the set of indices $i$ such that 
$\widetilde{q}_i=1$. This set specifies which $wt(\widetilde{q})$ coins of 
the $N$ coins are placed on the two pans. Let $P_{I(\widetilde{q})}$ be 
the set of all partitions of the set $I(\widetilde{q})$ of size $wt(\widetilde{q})$ 
($=$ even by Lemma \ref{parity-restricted}) into two sets of size $wt(\widetilde{q})/2$. 
Note that each partition $(Y,\overline{Y})$ in $P_{I(\widetilde{q})}$ specifies 
how to split the~$wt(\widetilde{q})$ coins in half to place them on the left and right pans, 
and can be identified with the corresponding query $q$ to the B-oracle. 
Finally, let $\chi(Y,\overline{Y})$ be the answer 
for the query $(Y,\overline{Y})\in P_{I(\widetilde{q})}$ to the B-oracle.

\

\noindent
{\bf Algorithm $Find^*(k)$.}

1. Prepare $N$ qubits $|0\rangle^{\otimes N}$ in a register ${\sf R}$, 
and apply a unitary transformation~$W$ of Lemma \ref{parity-restricted} to them. 
Then, we have the state
$\frac{1}{\sqrt{2^{N-1}}}\sum_{\widetilde{q}\in Q_{even}}|\widetilde{q}\rangle_{\sf R}$.

2. For each superposed $\widetilde{q}$, implement Steps 2.1--2.4 
on a register ${\sf R}'$ using $\widetilde{q}$ as a control part. 

\hspace{0.5cm} 
2.1. Apply a unitary transformation ${\cal A}_{\widetilde{q}}$ 
to the initial state $|0\rangle$ on ${\sf R}'$ to create a quantum state 
$
{\cal A}_{\widetilde{q}}|0\rangle:=
\frac{1}{\sqrt{|P_{I(\widetilde{q})}|}}\sum_{(Y,\overline{Y})\in P_{I(\widetilde{q})}}
|Y,\overline{Y}\rangle_{{\sf R}'},
$ 
which represents a uniform superposition of all partitions $(Y,\overline{Y})$ in $P_{I(\widetilde{q})}$.  
Then, the current state~is 
\begin{align*}
|\xi_{2,1}\rangle
&= \sum_{\widetilde{q}\in Q_{even}}|\widetilde{q}\rangle_{\sf R}
\sum_{(Y,\overline{Y})\in P_{I(\widetilde{q})}}\gamma\alpha|Y,\overline{Y}\rangle_{{\sf R}'}\\
&=\!\!\!\!\!\! 
\sum_{\widetilde{q}\in Q_{even}\cap Q_e}
|\widetilde{q}\rangle_{\sf R}
\sum_{(Y,\overline{Y})\in P_{I(\widetilde{q})}}\!\!\!\!
\gamma\alpha|Y,\overline{Y}\rangle_{{\sf R}'}
+
\sum_{\widetilde{q}\in Q_{even}\cap Q_o}
|\widetilde{q}\rangle_{\sf R}
\sum_{(Y,\overline{Y})\in P_{I(\widetilde{q})}}\!\!\!\!
\gamma\alpha|Y,\overline{Y}\rangle_{{\sf R}'}
\end{align*}
where $Q_{e}$ (resp.~$Q_{o}$) denotes the set of all $\widetilde{q}$'s 
such that $M(\widetilde{q},x)$ includes an even (resp.~odd) number of $1$'s. 
Also, $\gamma=1/\sqrt{2^{N-1}}$ and $\alpha=1/\sqrt{|P_{I(\widetilde{q})}|}$. 

\hspace{0.5cm}
2.2. Let $\overline{\chi}$ be the Boolean function defined by 
$\overline{\chi}(Y,\overline{Y})=1$ if and only if 
$\chi(Y,\overline{Y})=0$ (that is, the scale is balanced).
%the answer for the query $(Y,\overline{Y})$ to the B-oracle, $\chi(Y,\overline{Y})$, 
%is $0$ (that is, the scale is balanced). 
Then, under the above ${\cal A}_{\widetilde{q}}$ and $\overline{\chi}$, run the amplitude amplification algorithm 
${\bf QSearch}({\cal A}_{\widetilde{q}},\overline{\chi})$ 
when the initial success probability of ${\cal A}_{\widetilde{q}}$ 
is unknown (Theorem 3 in \cite{BHMT02}). 
%while the number of repetitions of the Grover-like subroutine (see the proof of Lemma \ref{length}) is at most $c_0k^{1/4}$, where $c_0$ is a large constant. 
Here ``success'' means the scale is balanced and hence we use~$\overline{\chi}$, not~$\chi$, in {\bf QSearch}. 
Then we obtain a state in the form of 
\[
|\xi_{2,2}\rangle 
=\!\!\!\!\!\!
\sum_{\widetilde{q}\in Q_{even}\cap Q_{e}}|\widetilde{q}\rangle_{{\sf R}}\!\!\!\!\!\!
%\left(\frac{1}{\sqrt{|P_{I(\widetilde{q})}|}}
\sum_{(Y,\overline{Y})\in P_{I(\widetilde{q})} }\!\!\!\!
\gamma \beta_{Y} |Y,\overline{Y},g_{Y}\rangle_{{\sf R'}}
+\!\!\!\!\!\!
\sum_{\widetilde{q}\in Q_{even} \cap Q_{o} }|\widetilde{q}\rangle_{{\sf R}}\!\!\!\!\!\!
%\left(\frac{1}{\sqrt{|P_{I(\widetilde{q})}|}}
\sum_{(Y,\overline{Y}) \in P_{I(\widetilde{q})} }\!\!\!\!
\gamma \alpha |Y,\overline{Y},g_{Y}\rangle_{{\sf R'}} 
\]
where $|g_{Y}\rangle$ is the garbage state. Note that, in the first term, 
the amplitudes $\beta_{Y}$ such that $\overline{\chi}(Y,\overline{Y})=1$ are now large 
by amplitude amplification while 
the second term does not change since the scale is always tilted. 

\hspace{0.5cm} 2.3. If Step 2.2 finds a ``solution,'' i.e., 
a partition $(Y,\overline{Y})$ such that $\overline{\chi}(Y,\overline{Y})=1$, 
then do nothing. 
%(and then $|\widetilde{q}\rangle$ is transformed into $|\widetilde{q}\rangle$). 
Otherwise, flip the phase (and then the phase is kick-backed into ${\sf R}$). 
%and $|\widetilde{q}\rangle_{\sf R}$ is transformed into $-|\widetilde{q}\rangle_{\sf R}$).
Notice that when $M(\widetilde{q},x)$ includes an odd number of $1$'s, 
the phase is always flipped, while when it includes an even number of $1$'s, 
the phase is not flipped with high amplitude. Now the current state is
\begin{align*}
|\xi_{2,3}\rangle &= \!\!\!\!\!\!\!\!
\sum_{\widetilde{q}\in Q_{even}\cap Q_e}\!\!\!\!\!\!|\widetilde{q}\rangle_{{\sf R}}\!\!\!\!
\sum_{(Y,\overline{Y})\in P_{I(\widetilde{q})}}\!\!\!\!\!\!\!\!
\gamma \beta_{Y} (-1)^{\chi(Y,\overline{Y})}|Y,\overline{Y},g_{Y}\rangle_{{\sf R'}}-\!\!\!\!\!\!\!\!\!\!\!
\sum_{\widetilde{q}\in Q_{even}\cap Q_o}\!\!\!\!\!\!
|\widetilde{q}\rangle_{{\sf R}}\!\!\!\!
\sum_{(Y,\overline{Y})\in P_{I(\widetilde{q})}}\!\!\!\!\!\!\!\!\gamma\alpha |Y,\overline{Y},g_{Y}\rangle_{{\sf R'}}\\
&= \!\!\!\!\!\!
\sum_{\widetilde{q}\in Q_{even}\cap Q_e}\!\!\!\!
|\widetilde{q}\rangle_{{\sf R}}
\sum_{(Y,\overline{Y})\in P_{I(\widetilde{q})}}
\!\!\!\! \!\!\gamma\beta_{Y} |Y,\overline{Y},g_{Y}\rangle_{{\sf R'}}
-\!\!\!\!\!\!\!\!\!\sum_{\widetilde{q}\in Q_{even}\cap Q_o} \!\!\!\!
|\widetilde{q}\rangle_{{\sf R}}
\sum_{(Y,\overline{Y})\in P_{I(\widetilde{q})}}
\!\!\!\!\!\! \gamma\alpha |Y,\overline{Y},g_{Y}\rangle_{{\sf R'}} -2\!\!\!\!\!\!\!\sum_{\widetilde{q}\in Q_{even}\cap Q_e}\!\!\!\!
|\widetilde{q}\rangle_{{\sf R}}
|err_{\widetilde{q}}\rangle_{{\sf R}'}
%\\&=& 
\ignore{
\sum_{\widetilde{q}\in Q_{even}}
(-1)^{\widetilde{q}\cdot x}|\widetilde{q}\rangle_{{\sf R}}
\sum_{(Y,\overline{Y})\in P_{I(\widetilde{q})}}\beta_{Y} |Y,\overline{Y},g_{Y}\rangle_{{\sf R'}}
}
\end{align*}
where $|err_{\widetilde{q}}\rangle_{{\sf R}'}
=\sum_{(Y,\overline{Y})\in P_{ I(\widetilde{q}) } :\chi(Y,\overline{Y})=1}
\gamma\beta_{Y} |Y,\overline{Y},g_{Y}\rangle_{{\sf R'}}$.

\hspace{0.5cm} 2.4. Reverse the quantum transformation done in Steps 2.1 and 2.2. 
Notice that the reversible transformation is done on ${\sf R}'$ in parallel for each $\widetilde{q}$ 
while the contents of ${\sf R}$ does not change since it is the control part. 
Therefore, the state becomes  
\begin{align*}
|\xi_{2,4}\rangle
&=
\frac{1}{\sqrt{2^{N-1}}}\!\!
\sum_{\widetilde{q}\in Q_{even}\cap Q_{e}}\!\!\!\!\!\!\!\!
|\widetilde{q}\rangle_{{\sf R}}|0\rangle_{{\sf R}'}
-
\frac{1}{\sqrt{2^{N-1}}}\!\!
\sum_{\widetilde{q}\in Q_{even}\cap Q_{o}}\!\!\!\!\!\!\!\!
|\widetilde{q}\rangle_{{\sf R}}|0\rangle_{{\sf R}'}
-
2\!\!\!\!\!\!\!\!\sum_{\widetilde{q}\in Q_{even}\cap Q_e}\!\!\!\!\!\!\!\!
|\widetilde{q}\rangle_{{\sf R}}
|err'_{\widetilde{q}}\rangle_{{\sf R}'}\\
&= 
\frac{1}{\sqrt{2^{N-1}}}\sum_{\widetilde{q}\in Q_{even}}
(-1)^{\widetilde{q}\cdot x}|\widetilde{q}\rangle_{{\sf R}}|0\rangle_{{\sf R}'}
-2\sum_{\widetilde{q}\in Q_{even}\cap Q_e}
|\widetilde{q}\rangle_{{\sf R}}
|err'_{\widetilde{q}}\rangle_{{\sf R}'}
\end{align*}
where $|err'_{\widetilde{q}}\rangle_{{\sf R}'}$ is the transformed state of $|err_{\widetilde{q}}\rangle_{{\sf R}'}$.
\ignore{(Recall that ${\bf QSearch}({\cal A},\chi)$ is based on the Grover subroutine consisting of ${\cal A}$, its inverse, the oracle transformation $S_\chi$ defined by  $S_\chi|Y,\overline{Y}\rangle=(-1)^{\chi(Y,\overline{Y})}|Y,\overline{Y}\rangle$, and the unitary transformation $U_{0}$ defined by $U_{0}|z\rangle=|z\rangle$ if $z\neq 0$ and $-|z\rangle$.) }

3. Apply $W^{-1}$ to the state in ${\sf R}$. Then we obtain a final state 
\[
|\xi_3\rangle=|x\rangle_{{\sf R}}|0\rangle_{{\sf R}'}-2W^{-1}\left(\sum_{\widetilde{q}\in Q_{even}\cap Q_e}
|\widetilde{q}\rangle_{{\sf R}}
|err'_{\widetilde{q}}\rangle_{{\sf R}'}\right).
\]
Then measure ${\sf R}$ in the computational basis. (End of Algorithm)

\

For justifying the correctness of $Find^*(k)$, it suffices to show that the squared magnitude of the second term of 
$|\xi_3\rangle$ is a small constant, say, $1/400$, since we then 
measure the desired value $x$ with probability at least $9/10$ 
(in fact, at least $(1-\sqrt{1/400})^2>9/10$). 
By the unitarity, its squared magnitude is equal to that of 
the last term of $|\xi_{2,3}\rangle$, that is, we want to evaluate the following value~$\epsilon$. 
\[
\epsilon:=4\left\|\sum_{\widetilde{q}\in Q_{even}\cap Q_e}
|\widetilde{q}\rangle_{{\sf R}}
|err_{\widetilde{q}}\rangle_{{\sf R}'}
\right\|^2=4\sum_{\widetilde{q}\in Q_{even}\cap Q_e}
|\widetilde{q}\rangle_{{\sf R}}
\left\| |err_{\widetilde{q}}\rangle_{{\sf R}'}
\right\|^2.
\]

\begin{lemma}\label{length}
$\epsilon$ is at most $1/400$. %(under a large constant $c_0$).
\end{lemma}

\begin{proof}
%For any $\widetilde{q}\in Q_e$, we can see that Step 2.3 does nothing with high probability. 
Consider an arbitrary $\widetilde{q}$ in $Q_{even}\cap Q_e$. 
When $M(\widetilde{q},x)$ includes $m$ $(\leq k)$ $1$'s (where $m$ is even), 
the state ${\cal A}_{\widetilde{q}}|0\rangle$ includes a partition $(Y,\overline{Y})$ 
such that $\overline{\chi}(Y,\overline{Y})=1$ with probability at least 
\[
p=\frac{\binom{m}{m/2}\binom{wt(\widetilde{q})-m}{(wt(\widetilde{q})-m)/2}}{\binom{wt(\widetilde{q})}{wt(\widetilde{q})/2}}=\Omega(1/\sqrt{m})=\Omega(1/\sqrt{k}).
\]
%(which is the initial success probability of ${\cal A}_{\widetilde{q}}$). 
By Theorem 3 in \cite{BHMT02}, it is guaranteed that, in the algorithm 
%\footnote{The basic idea of ${\bf QSearch}({\cal A}_{\widetilde{q}},\chi)$ was appeared in \cite{BBHT98}.} 
${\bf QSearch}({\cal A}_{\widetilde{q}},\overline{\chi})$, 
an expected number of applications of the Grover-like subroutine 
to find a ``solution,'' i.e., a partition $(Y,\overline{Y})$ such that $\overline{\chi}(Y,\overline{Y})=1$, 
is bounded by $O(1/\sqrt{p})=O(k^{1/4})$. 
The subroutine consists of (i) ${\cal A}_{\widetilde{q}}$,
(ii) its inverse, (iii) the transformation $O_{\overline{\chi}}$ defined by $O_{\overline{\chi}}|Y,\overline{Y}\rangle
=(-1)^{\overline{\chi}(Y,\overline{Y})}|Y,\overline{Y}\rangle$,
and (iv) the transformation~$U_{0}$ defined by $U_{0}|z\rangle=|z\rangle$
if $z\neq 0$ and $-|z\rangle$ if~$z=0$, where ${\cal A}_{\widetilde{q}}$ 
(and hence its inverse) and $U_{0}$ can be implemented without any query to the B-oracle, 
and~$O_{\overline{\chi}}$ can be implemented with one query to the B-oracle. 
Thus the expected number of queries to find a ``solution'' is $O(k^{1/4})$.
By setting the number of applications of the subroutine to $c_0k^{1/4}$ 
where $c_0$ is a large constant, Step 2.2 finds a ``solution'' 
with probability at least $1599/1600$. This means that for any $\widetilde{q}\in Q_{even}\cap Q_e$, 
$\sum_{(Y,\overline{Y})\in P_{I(\widetilde{q})}:\overline{\chi}(Y,\overline{Y})=0}\beta_{Y} 
|Y,\overline{Y},g_{Y}\rangle_{{\sf R'}}$ has squared magnitude at most $1/1600$.  
Recalling $\gamma=1/\sqrt{2^{N-1}}$ we have
\[
\epsilon=4\gamma^2\sum_{\widetilde{q}\in Q_{even}\cap Q_e}
\left\|\sum_{(Y,\overline{Y})\in P_{I(\widetilde{q})}:\overline{\chi}(Y,\overline{Y})=0}
\beta_{Y} 
|Y,\overline{Y},g_{Y}\rangle_{{\sf R'}}\right\|^2\leq 1/400. 
\] 
This completes the proof of Lemma \ref{length}. %\hfill$\Box$
\end{proof}

Finally, it is easy to see from the above proof that the query complexity of $Find^*(k)$ is $O(k^{1/4})$ 
since it makes $O(k^{1/4})$ queries in Step 2 and no queries in Steps 1 and 3. 

%\subsubsection{Exact Algorithm}\label{sec:exact_algo}

Now we consider the exact algorithm $Find(k)$. By the symmetric structure of algorithm $Find^*(k)$, 
the success probability of identifying $x$ correctly is independent of $x$ 
(recall that the oracle candidates are $\binom{N}{k}$ $N$-bit strings $x$ 
with Hamming weight $k$). Thus we can use the so-called exact amplitude amplification algorithm 
(Theorem 4 in \cite{BHMT02}) to convert it into the exact algorithm.

Here is the brief description of $Find(k)$. (see Appendix \ref{sec:find_k} for the details).
First, we implement $Find^*(k)$. As shown above, $Find^*(k)$
produces the correct output (i.e., $k$ false coins) with a constant probability ($\geq 9/10$) 
larger than $1/4$. Notice that we can make the success probability exactly $1/4$ 
by an appropriate adjustment. We need an algorithm for checking if the output is correct 
to amplify the success probability to $1$. Namely, an algorithm $Check$ needs to judge 
whether $k$ coins are indeed all false, which can be implemented classically 
in $O(\log k)$ weighings (as seen in Appendix \ref{sec:find_k}). 
Then we can implement the exact amplitude amplification: 
Like the $1/4$-Grover's algorithm \cite{BBHT98}, flip the phase if $Check$ judges that the output is correct, 
and apply the reflection about the state obtained after $Find^*(k)$. It is not difficult to see that 
$Find(k)$ always finds~$k$ false coins and the total complexity is $O(k^{1/4})$. 
Therefore, the proof of Theorem \ref{main_theorem} is completed. 
\end{proofof}

\section{Lower Bounds}\label{sec:lower}
\subsection{Basic Ideas}\label{sec:31}

In this section, we discuss the lower bound of finding $k$ false coins
from $N$ coins. We conjecture that the upper bound $O(k^{1/4})$ is tight but, 
unfortunately, we have not been able to show whether it is true or not. 
Instead, we show that if there would be an algorithm that improves 
the upper bound essentially, then it would have a completely different structure
from our algorithm.  

Before describing our results, we observe two properties of our algorithm $Find(k)$.
First, $Find(k)$ essentially uses only ``big pans,'' i.e., it always places 
at least $\Omega(N)$ coins on the pans, which is called the {\em big-pan property}. 
(The algorithm $Find^*(k)$ in Section \ref{one-fourth} uses ``small pans'' 
but it can be adapted with no essential change so that it works even if the size of pans must be big, 
as easily shown in Appendix \ref{sec:find_k}.)
Second, the B-oracle is always used in such a way that 
once the coins placed on the two pans are determined, the partition of them 
into the two pans is done uniformly at random, which is called the {\em random-partition property}. 
What we show in this section is that the current upper bound is best achievable   
for any algorithm that satisfies at least one of these two properties. 

For this purpose, we revisit one version of the (nonnegative) quantum adversary method, 
called {\it the strong weighted adversary method} in \cite{SS06}, due to Zhang~\cite{Zha04}. 
Let $f$ be a function from a finite set $S$ to another finite set $S'$. 
Recall that in a query complexity model, an input $x\in S$ is given as an oracle. 
An algorithm ${\cal A}$ would like to compute $f(x)$ while it can obtain 
the information about~$x$ by a unitary transformation 
$
O_x|q,a,z\rangle=|q,a\oplus\zeta(x;q),z\rangle,
$
where $|q\rangle$ is the register for a query string $q$ from a finite set $Q$, 
$|a\rangle$ is the register for the binary answer $\zeta(x;q)$, a function 
from $S\times Q$ to $\{0,1\}$, and $|z\rangle$ is the work register. 
Note that the adversary method usually assumes the so-called index oracle, 
namely $q$ is an integer $1 \leq i \leq N$ and $\zeta(x;q)$
is the $i$th bit ($0$ or~$1$) of $x\in\{0,1\}^N$. However, one can easily see that the above
generalization to $\zeta(x;q)$ requires no essential changes for its proof.
Thus Theorem 14 of \cite{Zha04} can be restated as follows:

\begin{lemma}\label{adv_boracle}
Let $w,w'$ denote a weight scheme as follows:
\begin{enumerate}
\item Every pair $(x,y)\in S\times S$ is assigned 
a nonnegative weight $w(x,y)=w(y,x)$ that satisfies $w(x,y)=0$ whenever $f(x)=f(y)$.

\item Every triple $(x,y,q)\in S\times S\times Q$ 
is assigned a nonnegative weight $w'(x,y,q)$ that satisfies $w'(x,y,q)=0$ 
whenever $\zeta(x;q)=\zeta(y;q)$ or $f(x)=f(y)$, and $w'(x,y,q)w'(y,x,q)\geq w^2(x,y)$ 
for all $x,y,q$ such that $\zeta(x;q)\neq \zeta(y;q)$ and $f(x)\neq f(y)$.
\end{enumerate}
For all $x,q$, let $\mu(x)=\sum_y w(x,y)$ and $\nu(x,q)=\sum_y w'(x,y,q)$. 
Then, the quantum query complexity of $f$ is at least
\[
\Omega\left(
\max_{w,w'}\min_{\substack{x,y,q:\ w(x,y)>0,\\ \zeta(x;q)\neq \zeta(y;q)} } 
\sqrt{\frac{\mu(x)\mu(y)}{\nu(x,q)\nu(y,q)}}
\right).
\]
\end{lemma}

\subsection{Big Pan Lower Bounds}\label{sec:32}
First, we show that our upper bound is tight under the big-pan property.
In what follows, $L\geq l$ denotes the restriction that at least $l$ coins 
must be placed on the pans whenever the balance is used. 

\begin{theorem}\label{bigpan}
If $L\geq l$, we need $\Omega((lk/N)^{1/4})$ weighings to find $k$ false coins. 
In particular, $\Omega(k^{1/4})$ weighings are necessary if there is some constant $c$ 
such that $L\geq N/c$.
\end{theorem} 

\begin{proof} Let $l=N/d$. Then the lower bound we should show is $\Omega((k/d)^{1/4})$.
We can assume that $d\leq k/3$ (otherwise, the lower bound becomes trivial). 
To use Lemma \ref{adv_boracle}, let $S=\{x\in\{0,1\}^N\mid wt(x)=k\}$,
$Q=Q_{\geq N/d}:=\bigcup_{l\geq N/d}Q_l$, $\zeta(x;q)=\chi(x;q)$, and $f(x)=x$. 
Our weight scheme is as follows: 
Let $w(x,y)=1$ for any pair $(x,y)\in S\times S$ such that $x\neq y$, 
and let $w'(x,y,q)=1$ for all $(x,y,q)\in S\times S\times Q_{\geq N/d}$ 
such that $\chi(x;q)\neq \chi(y;q)$ and $x\neq y$. 
It is easy to check that this satisfies the condition of a weight scheme. 
Then, for any $x$, we have $\mu(x)=\sum_y w(x,y)=\binom{N}{k}-1$. 
We need to evaluate $\nu(x,q)\nu(y,q)$ 
for pairs $(x,y)$ such that $\chi(x;q)=1$ and $\chi(y;q)=0$ or 
$\chi(x;q)=0$ and $\chi(y;q)=1$. Fix $q\in Q_{\geq N/d}$ arbitrarily
and assume that $q\in Q_{N/c}$ where $c\leq d$.  
When $\chi(x;q)=1$ (i.e., the scale is tilted for query $q$ when $x$ is the input), 
notice that $\nu(x,q)=\sum_{y}w'(x,y,q)$ is the number of all $y$'s 
such that the scale is balanced when $N/c$ coins are placed 
on each of the two pans according to $q$. 
Therefore, by summing up all the cases such that those 
$N/c$ coins include $m$ false ones, 
\[
\nu(x,q)=\gamma(N,k,c):=\sum_{m=0}^{k/2}\binom{N/c}{m}^2\binom{(1-2/c)N}{k-2m}.
\] 
Since $\chi(y;q)=0$, we have $\nu(y,q)=\sum_x w'(x,y,q)=
\binom{N}{k}-\gamma(N,k,c)$ by counting all $x$'s such that the scale is titled. 
Then the product $\nu(x,q)\nu(y,q)$ is $\gamma(N,k,c)\left(\binom{N}{k}-\gamma(N,k,c)\right)$. 
Similarly, when $\chi(x;q)=1$ we can see that the product is also 
$\gamma(N,k,c)\left(\binom{N}{k}-\gamma(N,k,c)\right)$. 
By Lemma \ref{adv_boracle} the quantum query complexity of our problem 
is at least
\begin{equation}\label{eq091024-1}
\Omega\left(
\min_{c:\ c\leq d} \sqrt{\frac{(\binom{N}{k}-1)^2}{\gamma(N,k,c)(\binom{N}{k}-\gamma(N,k,c))}}
\right)
= \Omega\left(\min_{c:\ c\leq d}\sqrt{\frac{\binom{N}{k}}{\gamma(N,k,c)}}
\right).
\end{equation}
Then, we can show the following lemma. 
%(see Appendix \ref{sec:balance_sum} for the proof). 

\begin{lemma}\label{balance_sum}
$\gamma(N,k,c)/\binom{N}{k}=O(\sqrt{c/k})$ for any $2\leq c\leq d$ $(\leq k/3)$.
\end{lemma}

%\subsection{Proof of Lemma \ref{balance_sum}}\label{sec:balance_sum}

\begin{proof}
Note that $\gamma(N,k,c)/\binom{N}{k}$ means the probability that
the scale is balanced when $N/c$ coins ($N$ coins include $k$ false ones)
are randomly placed on each of the two pans, and hence
its value decreases as $c$ approaches to $2$.
So, it suffices to prove the lemma for $c\geq 4$. 
 
Let us denote each term in the sum $\gamma(N,k,c)$
by $t(m)=\binom{N/c}{m}^2\binom{(1-2/c)N}{k-2m}$ for $m=0,1,\ldots,k/2$.
We divide $\gamma(N,k,c)$ into the two parts, that is, 
we write $\gamma(N,k,c)= T_{>k/2c}+T_{\leq k/2c}$ 
where $T_{>k/2c}=\sum_{m:m>k/2c} t(m)$ and $T_{\leq k/2c}=\sum_{m:m\leq k/2c} t(m)$. 
For the proof, it suffices to show that both $T_{>k/2c}/\binom{N}{k}$ 
and $T_{\leq k/2c}/\binom{N}{k}$ are bounded by $O(\sqrt{c/k})$.
First we consider $T_{>k/2c}/\binom{N}{k}$. When $N/c$ coins are randomly placed on each of the pans,
let $E_1$ be the event that at least $k/2c$ false coins
are placed on the pans, and $E_2$ be the event that the scale is balanced.
Then, we can see that $T_{>k/2c}/\binom{N}{k}=\mathrm{Pr}[E_1\wedge E_2]$
which is at most $\mathrm{Pr}[E_2|E_1]=O(1/\sqrt{k/c})=O(\sqrt{c/k})$.  
Second we consider $T_{\leq k/2c}/\binom{N}{k}$. Let $r(m)=t(m+1)/t(m)$.
Note that $r(m)$ is monotone decreasing on $m$ since 
\begin{align*}
r(m) 
&=
\frac{\binom{N/c}{m+1}^2\binom{(1-2/c)N}{k-2(m+1)}}{\binom{N/c}{m}^2\binom{(1-2/c)N}{k-2m}}\\
&=\frac{(\frac{N}{c}-m)^2(k-2m)(k-2m-1)}{(m+1)^2((1-2/c)N-k+2m+1)((1-2/c)N-k+2m+2)}.
\end{align*}
Now we verify that $r(k/2c-1)>4$. In fact, since $c\leq k/3<2+k/2$, we have
\begin{equation}\label{eq0910-1}
(1-2/c)N-k+k/c<(1-2/c)(N-k/2-c)
\end{equation}
and
\begin{equation}\label{eq0910-2}
k-k/c-3\geq k(1-2/c).
\end{equation}
Thus we obtain
\begin{align*}
r(k/2c-1) &=
\frac{(1/c)^2(N-k/2-c)^2(k-k/c-2)(k-k/c-3)}{(k/2c)^2((1-2/c)N-k+k/c)((1-2/c)N-k+k/c-1)}\\
&> \frac{4(k-k/c-2)(k-k/c-3)}{k^2(1-2/c)^2}\ \ (\mbox{by Eq.(\ref{eq0910-1})})\\
&\geq 4\ \ (\mbox{by Eq.(\ref{eq0910-2})}). 
\end{align*}
These facts imply that
\[
T_{\leq k/2c}=\sum_{m:m\leq k/2c}t(m)<\left(1+1/4+(1/4)^2+\cdots\right)t(k/2c)=(4/3)t(k/2c),
\]
which is bounded by $(4/3)t(k/c)$ since $t(m)$ takes the maximum value when $m=k/c$. 
Calculating $t(k/c)/\binom{N}{k}$ using the Stirling formula $n!\sim \sqrt{2\pi n}(N/e)^N$,
we obtain
\begin{align*}
\frac{t(k/c)}{\binom{N}{k}}
&=
\frac{\binom{N/c}{k/c}^2\binom{(1-2/c)N}{(1-2/c)k}}{\binom{N}{k}}=
\frac{\frac{k!}{((\frac{k}{c})!)^2((1-\frac{2}{c})k)!}\cdot \frac{(N-k)!}{((\frac{N-k}{c})!)^2((1-\frac{2}{c})(N-k))!}
}{\frac{N!}{((\frac{N}{c})!)^2((1-\frac{2}{c})N)!}}\\
&\sim \frac{cN}{2\pi k(N-k)\sqrt{1-2/c}},
\end{align*}
which is bounded by $O(c/k)$ since $k\leq N/2$ and $c\geq 4$.
Thus, the sum $T_{\leq k/2c}/\binom{N}{k}$ is bounded by $O(c/k)=O(\sqrt{c/k})$.
From the above, we obtain $\gamma(N,k,c)/\binom{N}{k}=O(\sqrt{c/k})$. 
\end{proof}

Now Lemma \ref{balance_sum} implies the desired bound $\Omega((k/d)^{1/4})$ 
by Eq.(\ref{eq091024-1}), and hence the proof of Theorem~\ref{bigpan} is completed. %\hfill$\Box$
\end{proof}

On the contrary, we can show that any algorithm that uses only ``small
pans'' also needs $\Omega(k^{1/4})$ queries (Theorem \ref{smallpan}). 
For instance, we cannot break the current bound $k^{1/4}$ by any algorithm that places $O(N/k)$ coins on the
pans. (Notice that the pan includes only a constant number of false
coins with high probability in this case and therefore we can achieve a better success
probability for the even false-coin case, but at the same time, we
cannot use a wide range of superpositions). Moreover, we can obtain
another lower bound for the case where ``big pans'' and ``small pans''
are both available but ``medium pans'' are not (Theorem
\ref{bigpan_generalized}).  Unfortunately one can see that there is
still a gap between the sizes of the big pans and small pans even for
a weakest nontrivial ($\omega(1)$) lower bound. See Appendix \ref{sec:pan_size} 
for the details of these results.

\subsection{Lower Bounds for the Quasi B-Oracle}\label{sec:33}

Second, we show that our upper bound is tight under the random-partition property. 
Notice that in this case, if the coins include an odd number of false ones, 
then the scale is always tilted, and if the coins include an even number~(=$m$) 
of false ones, the scale will be balanced with probability $1/\sqrt{m}$. 
Thus in order to show a lower bound, we need to generalize the adversary method
that works for such ``stochastic'' oracles:  Now $\zeta(x;q)$ is a random variable
and the stochastic version of $O_x$, denoted by $\widetilde{O}_x$, is defined as 
(we should be careful not to lose its unitarity):
\[
\widetilde{O}_x|q,a,z\rangle=
\sqrt{\mathrm{Pr}[\zeta(x;q)=0]}|q,a,z\rangle +(-1)^a\sqrt{\mathrm{Pr}[\zeta(x;q)=1]}
|q,a\oplus 1,z\rangle.
\]

Now Lemma \ref{adv_boracle} changes to the following: 

\begin{lemma}\label{general_lower}
Let $w,w'$ denote a weight scheme as Lemma \ref{adv_boracle} except replacing Condition 2 to
\begin{itemize}
\item[2'] Every triple $(x,y,q)\in S\times S\times Q$ is assigned a nonnegative weight $w'(x,y,q)$ 
that satisfies $w'(x,y,q)=0$ whenever $\mathrm{Pr}[\zeta(x;q)=\zeta(y;q)]=1$ or $f(x)=f(y)$,
and $w'(x,y,q)w'(y,x,q)\geq w^2(x,y)$ for all $x,y,q$ such that $\mathrm{Pr}[\zeta(x;q)\neq\zeta(y;q)]>0$
and $f(x)\neq f(y)$. 
\end{itemize}
%For all $x,q$, let $\mu(x)=\sum_y w(x,y)$ and $\nu(x,q)=\sum_{y} w'(x,y,q)$.
Then, the quantum query complexity of $f$ is at least
\[
\Omega\left(
\max_{w,w'} \min_{\substack{x,y,q:\ w(x,y)>0, \\ \mathrm{Pr}[\zeta(x;q)\neq \zeta(y;q) ]>0 }}
\sqrt{\frac{\mu(x)\mu(y)}{\nu(x,q)\nu(y,q)}}\frac{1}{\sqrt{P_{01,q}}+\sqrt{P_{10,q}}}
\right),
\]
where $P_{ab,q}=\mathrm{Pr}[\zeta(x;q)=a]\mathrm{Pr}[\zeta(y;q)=b]$.
\end{lemma}

\begin{proof}
%\subsection{Proof of Lemma \ref{general_lower}}\label{sec:general_lower}
The proof follows that of \cite[Theorem 14]{Zha04} essentially; 
in the following we mainly describe the difference.  
Assume that there is a $T$-query quantum algorithm ${\cal A}$ computing $f$ with high probability. 
Note that the initial state of ${\cal A}$ is $|\psi_x^0\rangle=|0\rangle$ 
for any input $x$. The final state for input $x$ can be written as 
$|\psi_x^{T}\rangle=U_{T-1}\widetilde{O}_x \cdots U_1\widetilde{O}_xU_0|0\rangle$ 
for some unitary transformations $U_0,\ldots,U_{T-1}$.
Since ${\cal A}$ computes $f$ with high probability, there is some constant $\epsilon<1$ 
such that $|\langle\psi_x^T|\psi_y^T\rangle|\leq \epsilon$ for any $x$ and $y$ with $f(x)\neq f(y)$. 
Let $|\psi_x^k\rangle=U_{k-1}\widetilde{O}_x\cdots U_1\widetilde{O}_x U_0|0\rangle$. 
For any $x$ and $y$ with $f(x)\neq f(y)$, we can 
represent
\[
|\psi_x^{k-1}\rangle = \sum_{q,a,z}\alpha_{q,a,z}|q,a,z\rangle,\quad\quad\quad
|\psi_y^{k-1}\rangle = \sum_{q,a,z}\beta_{q,a,z}|q,a,z\rangle.
\]
After querying to the oracle, we have
\begin{align*}
\widetilde{O}_x|\psi_x^{k-1}\rangle &= 
\sum_{q,a,z}\alpha_{q,a,z}( {\sqrt{\mathrm{Pr}[\zeta(x;q)=0]}
|q,a,z\rangle + (-1)^a\sqrt{\mathrm{Pr}[\zeta(x;q)=1 ]}|q,a\oplus 1,z\rangle} )\\
  &= \sum_{q,a,z}( \sqrt{\mathrm{Pr}[\zeta(x;q)=0] }
\alpha_{q,a,z} +(-1)^{a\oplus 1}\sqrt{\mathrm{Pr}[\zeta(x;q)=1] }\alpha_{q,a\oplus 1,z})
|q,a,z\rangle,\\
\widetilde{O}_y|\psi_y^{k-1}\rangle 
\ignore{
&= \sum_{q,a,z} \beta_{q,a,z} \{ \sqrt{\mathrm{Pr}[\zeta(y;q)=0]}
|q,a,z\rangle + (-1)^a\sqrt{\mathrm{Pr}[\zeta(y;q)=1]}|q,a\oplus 1,z \rangle \} \\
}
 &= \sum_{q,a,z} ( \sqrt{\mathrm{Pr}[\zeta(y;q)=0]}\beta_{q,a,z} 
+ (-1)^{a \oplus 1} \sqrt{\mathrm{Pr}[\zeta(y;q)=1]}\beta_{q,a\oplus 1,z} ) 
|q,a,z\rangle.
\end{align*}
Hence we have (recall that $P_{ab,q}:=\mathrm{Pr}[\zeta(x;q)=a]\mathrm{Pr}[\zeta(y;q)=b]$):
\begin{align*}
\langle \psi_x^{k} | \psi_y^{k} \rangle 
 &= \sum_{q,a,z}\sqrt{P_{00,q}}\alpha^*_{q,a,z}\beta_{q,a,z}
+\sum_{q,a,z}\sqrt{P_{11,q}}\alpha^*_{q,a\oplus 1,z} \beta_{q,a\oplus 1,z}\\
& \ \ + \sum_{q,a,z}(-1)^{a\oplus 1}\sqrt{P_{01,q}}\alpha^*_{q,a,z}\beta_{q,a\oplus 1,z} 
+ \sum_{q,a,z}(-1)^{a\oplus 1}\sqrt{P_{10,q}}\alpha^*_{q,a\oplus 1,z} \beta_{q,a,z}\\
&= \sum_{q,a,z}\sqrt{P_{00,q}}\alpha^*_{q,a,z}\beta_{q,a,z}
+\sum_{q,a,z}\sqrt{P_{11,q}}\alpha^*_{q,a,z} \beta_{q,a,z}\\
& \ \ + \sum_{q,a,z}(-1)^{a\oplus 1}\sqrt{P_{01,q}}\alpha^*_{q,a,z}\beta_{q,a\oplus 1,z} 
+ \sum_{q,a,z}(-1)^{a}\sqrt{P_{10,q}}\alpha^*_{q,a,z} \beta_{q,a\oplus 1,z}.
\end{align*}
 On the contrary, 
\[
\langle \psi_x^{k-1} | \psi_y^{k-1} \rangle = 
\sum_{q,a,z}\alpha^*_{q,a,z} \beta_{q,a,z}.
\]
Thus the difference between $\langle \psi_x^{k-1} | \psi_y^{k-1} \rangle$ 
and $\langle \psi_x^{k} | \psi_y^{k} \rangle$ is 
\begin{align*}
\langle \psi_x^{k-1} | \psi_y^{k-1} \rangle - \langle \psi_x^k | \psi_y^k \rangle 
&= \!\!\!\!\!\!\!\!\!\!\!\!\!\!\!\!
\sum_{q,a,z:\mathrm{Pr}[\zeta(x;q)\neq\zeta(y;q)]>0 } 
\left[ (1-\sqrt{P_{00,q}}-\sqrt{P_{11,q}} )\alpha^*_{q,a,z} \beta_{q,a,z} \right. \\
& \left. + (-1)^{a}
( \sqrt{P_{01,q}}\alpha^*_{q,a,z} \beta_{q,a\oplus 1,z} - \sqrt{P_{10,q}}\alpha^*_{q,a,z} 
\beta_{q,a\oplus 1,z} ) \right]
\end{align*}
since $\mathrm{Pr}[\zeta(x;q)=\zeta(y;q)]=1$, 
that is, $P_{00,q}+P_{11,q}=1$ implies that $P_{00,q}=1$ or $P_{11,q}=1$. 
By the triangle inequality,
\begin{align*}
1-\epsilon
&\leq 1-|\langle \psi_x^T | \psi_y^T \rangle | \leq \sum_{k=1}^{T}|\langle \psi_x^{k-1} | \psi_y^{k-1} \rangle - \langle \psi_x^k | \psi_y^k \rangle |\\
\ignore{ & \leq & \sum_{q,a,z:\mathrm{Pr}[\zeta(x;q)\neq\zeta(y;q)]>0} \left[ (1-\sqrt{P_{00,q}}-\sqrt{P_{11,q}})|\alpha_{q,a,z}||\beta_{q,a,z}| + \sqrt{P_{01,q}}|\alpha_{q,a,z}||\beta_{q,a\oplus 1,z}|+\sqrt{P_{10,q}}|\alpha_{q,a\oplus 1,z}||\beta_{q,a,z}|  \right]\\}
& \leq \sum_{k=1}^T  \sum_{\substack{q,a,z\\ \mathrm{Pr}[\zeta(x;q) \neq \zeta(y;q)]>0}} 
\!\!\!\!\!\!\!\!\!\!\!\!\!\! \left[ (1-\sqrt{P_{00,q}}-\sqrt{P_{11,q}} )|
\alpha_{q,a,z}||\beta_{q,a,z}| + (\sqrt{P_{01,q}}+\sqrt{P_{10,q}} )
|\alpha_{q,a,z}||\beta_{q,a\oplus 1,z}| \right]\\
& \leq \sum_{k=1}^T \sum_{\substack{q,a,z\\ \mathrm{Pr}[\zeta(x;q)\neq\zeta(y;q)]>0} } 
\!\!\!\!\!\!\!\!\!\!\!\!\!\!\![ (\sqrt{P_{01,q}}+\sqrt{P_{10,q}} )(|\alpha_{q,a,z}||\beta_{q,a,z}| 
+ |\alpha_{q,a,z}||\beta_{q,a\oplus 1,z}|) ].
\end{align*}

The remaining part is completely similar to the proof of \cite[Theorem 14]{Zha04}. 
Summing up the inequalities for all $(x,y)\in S\times S$ with weight $w(x,y)$, 
we have $(1-\epsilon)\sum_{x,y}w(x,y)\leq 2T\frac{1}{\sqrt{A}} \sum_{x,y}w(x,y)$ 
where 
\[
A=\min_{\substack{x,y,q:\ w(x,y)>0 \\ \mathrm{Pr}[\zeta(x;q) \neq \zeta(y;q)]>0}}
\frac{\mu(x)\mu(y)}{\nu(x,q)\nu(y,q)}\frac{1}{(\sqrt{P_{01,q}}+\sqrt{P_{10,q}})^2}.
\]
Therefore, we obtain $T=\Omega(\sqrt{A})$ and hence the proof is completed.
\end{proof}

Now we define the stochastic version of our B-oracle by setting 
\[
\mathrm{Pr}[\zeta(x;q)=0]=\left\{
\begin{array}{ll}
0          & (\mbox{if}\ wt(x\wedge q)\ \mbox{is odd})\\
\sqrt{1/wt(x\wedge q)} & (\mbox{if}\ wt(x\wedge q) \mbox{ is positive
and even})\\
1          & (\mbox{if}\ wt(x\wedge q)=0),
\end{array}
\right.
\]
where $x$ and $q$ are $N$-bit strings, and $x\wedge q$ is the $N$-bit string 
obtained by the bitwise AND of $x$ and $q$.  We call this oracle
the {\em quasi B-oracle} and one can see that it simulates the
B-oracle with the random-partition property. 
% mentioned at the beginning of Section \ref{sec:31}.  
Now we are ready to give the upper and lower bounds for the query complexity of this quasi B-oracle. 
Assume that $wt(x)=k$. The upper bound is easy by modifying Theorem \ref{main_theorem} 
so that Step 2 in $Find^*(k)$ can be replaced with $O(k^{1/4})$ repetitions of the quasi B-oracle. 

\begin{theorem}\label{quasi_upper} 
There is an $O(k^{1/4})$-query quantum algorithm to find $x$ using the quasi B-oracle. 
\end{theorem}

On the contrary, we can obtain the tight lower bound by using Lemma \ref{general_lower}. 
The weight scheme contrasts with that of Theorem \ref{bigpan}; 
$w(x,y)$ is nonzero only if the Hamming distance between $x$ and $y$ is $2$. 
%See Appendix \ref{sec:quasi_lower} for the proof. 

\begin{theorem}\label{quasi_lower}
Any quantum algorithm with the quasi B-oracle needs $\Omega(k^{1/4})$ queries to find $x$.
\end{theorem}

\begin{proof}
%\subsection{Proof of Theorem \ref{quasi_lower}}\label{sec:quasi_lower}
First we define a weight scheme. Let $S=\{x\in\{0,1\}^N\mid wt(x)=k\}$  
and $f(x)=x$. In what follows, we assume that $wt(q)=l$ for the $q$ 
that provides the minimum value of the formula of Lemma \ref{general_lower} 
and show that the theorem holds for an arbitrary $l \leq N$. 
%In what follows, we assume that $wt(q)=l$ for arbitrarily taken $l\leq N$. 
For any $(x,y)\in S\times S$, let $w(x,y)=1$ if ${d}(x,y)=2$ and $0$ otherwise.
We must satisfy $w'(x,y,q)=0$ for any different $x,y$ such that ${d}(x,y)\neq 2$ or 
$\mathrm{Pr}[\zeta(x;q)=\zeta(y;q)]=1$, which implies $wt(x\wedge q)=wt(y\wedge q)$. 
Thus we let $w'(x,y,q)\neq 0$ only if ${d}(x,y)=2$ and $wt(x\wedge q)=wt(y\wedge q)\pm 1$.
Define $w'(x,y,q)$ as a function of $wt(x\wedge q)=m_1$ and $wt(y\wedge q)=m_2$, 
and thus denote it by $w'(x,y,q)=w'(m_1,m_2)$. Then $w'(m_1,m_2)$ is taken as
\[
w'(m_1,m_2)=
\left\{
\begin{array}{ll}
\frac{2m(N-k-l+2m)}{(l-2m+1)(k-2m+1)}  &\quad \mbox{if }(m_1,m_2)=(2m-1,2m) \\
\frac{(l-2m+1)(k-2m+1)}{2m(N-k-l+2m)}  &\quad \mbox{if }(m_1,m_2)=(2m,2m-1) \\
1  &\quad \mbox{if }(m_1,m_2)=(2m,2m+1),(2m+1,2m)\\
0  &\quad \mbox{otherwise}.
\end{array}
\right.
\]
It can be easily seen that $w,w'$ is a weight scheme. 
Now we evaluate the lower bound under this weight scheme.
Clearly, $\mu(x)=\mu(y)=k(N-k)$. For evaluating $\nu(x,q)\nu(y,q)$, 
we consider only the case where $m_1=wt(x\wedge q)=2m$ and $m_2=wt(y\wedge q)=2m-1$ 
(the other cases such as $m_1=2m$ and $m_2=2m+1$ can be similarly analyzed).
In this case, we have
\begin{align*}
\nu(x,q)
&=2m(N-l-k+2m)w'(2m,2m-1)+(k-2m)(l-2m)w'(2m,2m+1)\\
&\leq (l-2m+1)(k-2m+1)+(k-2m)(l-2m) \\ 
&\leq 2(l-2m+1)(k-2m+1),\\
\nu(y,q)
&=(2m-1)(N-k-l+2m-1)w'(2m-1,2m-2)\\
&\ + (k-2m+1)(l-2m+1)w'(2m-1,2m)\\
&\leq (2m-1)(N-k-l+2m-1) + 2m(N-k-l+2m) \\
&\leq 4m(N-k-l+2m).
\end{align*}
Note that since $\mathrm{Pr}[\zeta(x;q)=0]=\sqrt{1/2m}$ and $\mathrm{Pr}[\zeta(y;q)=1]=1$,
$P_{01,q}=1/\sqrt{2m}$ and $P_{10,q}=0$. Thus we have
\[
\frac{\mu(x)\mu(y)}{\nu(x,q)\nu(y,q)}
\frac{1}{(\sqrt{P_{01,q}}+\sqrt{P_{10,q}})^2}
=
\frac{k^2(N-k)^2\sqrt{2m}}{8m(l-2m+1)(k-2m+1)(N-k-l+2m)}.
\] 
This value is bounded below by $\Omega(k^{1/2})$ since 
$m\leq k/2$ and $l\leq N$. Now Lemma~\ref{general_lower} completes the proof.
\end{proof}

\

\noindent
{\bf Acknowledgements.} We are grateful to Mario Szegedy for 
directing our interest to the topic of this paper, and an anonymous 
referee for a helpful idea to improve the earlier upper bounds for 
general $k$ significantly. We are also grateful to Seiichiro Tani and Shigeru Yamashita 
for helpful discussions.   

%\vspace{-2mm}

\appendix

%\section{Appendix}

\section{Efficient Construction of Transformation $W$}\label{make_w}
It can be easily seen that our algorithm $Find^*(k)$ can be implemented 
in time polynomial in the length of the input except for a bit nontrivial task, 
constructing the transformation $W$. Precisely, $W$ is a unitary transformation that satisfies
$W|x\rangle=|\psi_x\rangle:=\frac{1}{\sqrt{2^{N-1}}}\sum_{\widetilde{q}\in Q_{even}}
(-1)^{\widetilde{q}\cdot x}|\widetilde{q}\rangle$ for any $x\in S_{<N/2}$.
We define a subset $S_{lh}$ of size $2^N/2$ as follows: $S_{lh}=S_{<N/2}$ if $N$ is odd,
or $S_{lh}=S_{<N/2}\cup\{x\in\{0,1\}^{N/2}\mid lex(x)\leq 2^{N/2}/2\}$
(where $lex(x)$ is the lexicographic order of $x$ in $\{0,1\}^{N/2}$) if $N$ is even.
Notice that $S_{lh}$ is a polynomial-time computable set.
Then the following algorithm implements $W$.
 
\
 
\noindent
{\bf Algorithm $A_W$.} Input: $|x\rangle$ such that $wt(x)<N/2$ in a register ${\sf S}$.
 
1. Create the quantum state $\frac{1}{\sqrt{2}}(|x\rangle+|\bar{x}\rangle)$ in ${\sf S}$
by Steps 1.1--1.3.
 
\hspace{0.5cm} 1.1. Prepare $\frac{1}{\sqrt{2}}(|0\rangle+|1\rangle)$ in a register ${\sf R}$.
 
\hspace{0.5cm} 1.2. If the content of ${\sf R}$ is $1$, flip all the $N$ bits in ${\sf S}$.
 
\hspace{0.5cm} 1.3. If the content of ${\sf S}$ is not in $S_{lh}$, flip the bit in ${\sf R}$.
 
2. Apply the Hadamard transform $H$ on ${\sf S}$.
 
3. Let ${\sf S}$ be the output.
 
\
 
It is easy to see that $A_W$ is implemented in polynomial time.
By Step 1.1, we have $\frac{1}{\sqrt{2}}|x\rangle_{\sf S}(|0\rangle+|1\rangle)_{\sf R}$.
After Step 1.2, the state becomes $\frac{1}{\sqrt{2}}
(|x\rangle_{\sf S}|0\rangle_{\sf R}+|\bar{x}\rangle_{\sf S}|1\rangle_{\sf R})$.
Step 1.3 transforms the state to
\[
\frac{1}{\sqrt{2}}(|x\rangle_{\sf S}|0\rangle_{\sf R}+|\bar{x}\rangle_{\sf S}|0\rangle_{\sf R})
=
\frac{1}{\sqrt{2}}(|x\rangle_{\sf S}+|\bar{x}\rangle_{\sf S})|0\rangle_{\sf R}.
\]
Finally, the state after Step 2 is
\[
H\left(\frac{1}{\sqrt{2}}(|x\rangle_{\sf S}+|\bar{x}\rangle_{\sf S})\right)|0\rangle_{\sf R}
=|\psi_x\rangle_{\sf S}|0\rangle_{\sf R}
\]
as shown in the proof of Lemma \ref{parity-restricted}.

\section{Algorithm $Find(k)$}\label{sec:find_k}

The exact algorithm $Find(k)$ is given as follows.

\

\noindent
{\bf Algorithm $Find(k)$.} 
Let $a~(\geq 9/10)$ be the success probability of $Find^*(k)$. 
Let~${\cal B}$ be the algorithm that uses a single qubit with initial state $|0\rangle$ 
and rotates it to $\sqrt{1-1/4a}|0\rangle+\sqrt{1/4a}|1\rangle$.
Notice that the probability that $Find^*(k)$ succeeds and ${\cal B}$ outputs $|1\rangle$ is exactly $1/4$.

(i) Run $Find^*(k)$ with initial state $|0\rangle_{\sf R}$ 
in the register ${\sf R}$ and obtain a candidate of $k$ false coins $X$ 
(in fact, the corresponding oracle), and also run ${\cal B}$ with initial state $|0\rangle_{\sf R'}$ 
in the register ${\sf R'}$. Let $U$ be the unitary transformation done in this step 
(that is, the state after this step is $U|0\rangle_{\sf R}|0\rangle_{\sf R'}$). 

(ii) Implement Steps (ii-1)--(ii-3) below.

\hspace{0.5cm}
(ii-1) Run algorithm $Check$, which will be described later, to check if $X$ is indeed the set of $k$ false coins.

\hspace{0.5cm} 
(ii-2) If $Check$ outputs YES and ${\cal B}$ outputs $|1\rangle$, flip the phase. 
Otherwise, do nothing. 

\hspace{0.5cm} 
(ii-3) Reverse the operation of Step (ii-1).

(iii) Apply the reflection about the state $U|0\rangle_{\sf R}|0\rangle_{\sf R'}$, i.e., 
$I-2U|0\rangle\langle 0|U^\dagger$, where $|0\rangle=|0\rangle_{\sf R}|0\rangle_{\sf R'}$, 
to the state.

(iv) Measure ${\sf R}$ in the computational basis.

\

By a geometric view (Theorem 4 in \cite{BHMT02}) similar to the Grover search where the fraction of correct solution(s) 
is $1/4$ \cite{BBHT98}, we can verify that $Find(k)$ succeeds with certainty.
In $Find(k)$, the ``solution'' is $|{X}\rangle_{\sf R}|1\rangle_{\sf R'}$
where ${X}$ is the $k$ false coins. Notice that Step (ii) implements the transformation
that changes $|{X}\rangle_{\sf R}|b\rangle_{\sf R'}$ to $-|{X}\rangle_{\sf R}|b\rangle_{\sf R'}$
if $({X},b)$ is the ``solution'' and $|{X}\rangle_{\sf R}|b\rangle_{\sf R'}$ otherwise.
The total complexity is the number of queries to run $Find^*(k)$ 
and its inverse three times (once for Step (i) and twice for Step (iii)) 
plus the number of queries to run $Check$ and its inverse. 
So, we obtain a query complexity of $O(k^{1/4})$ if $Check$ has a similar complexity. 

In fact, $Check$ needs only $O(\log k)$ queries, which is given as follows. 
For simplicity, we assume that $N$ is a multiple of $k+1$ and $k+1$ is a power of $2$ 
but the generalization is easy. (Note that the following algorithm satisfies 
the big-pan property. If we do not care the property, the algorithm can be simplified a lot.)

\

\noindent
{\bf Algorithm $Check$.}

Input: Two subsets of a set $X$ of $N$ coins, $X_1$ with size $k$ and $\overline{X}_1=X\setminus X_1$ with size $N-k$.

Output: YES iff the coins in $X_1$ are all false and the coins in $\overline{X}_1$ are all fair. 
 
1. Divide $\overline{X}_1$ into $k+1$ equal-sized subsets $Y_1,Y_2,\ldots,Y_{k+1}$ (recall the above assumption).

2. Let $L=Y_1$ and $R=Y_2$. For $i=1$ to $\log{(k+1)}$, repeat Steps 2.1--2.2.

\hspace{0.5cm} 2.1. Check if $L$ and $R$ are balanced by Steps 2.1.1--2.1.3. 

\hspace{0.8cm} 2.1.1. Construct arbitrarily two subsets $L'$ and $R'$ of size $N/4-|L|$ ($=N/4-|R|$) 
 from $X\setminus(X_1\cup L\cup R)$ (this is possible since 
$|X\setminus(X_1\cup L\cup R)|\geq N-k-|L|-|R|\geq(N/4-|L|)+(N/4-|R|)$). 

\hspace{0.8cm} 2.1.2. Compare $L\cup L'$ and $R\cup R'$ by a scale. If it is tilted, output NO. 

\hspace{0.8cm} 2.1.3. Compare $R\cup L'$ and $L\cup R'$ by a scale. If it is tilted, output NO. 

\hspace{0.5cm} 2.2. Set $L:=L \cup R$ and $R:=\bigcup_{j=2^{i}+1}^{2^{i+1}} Y_j$. 

3. Output YES.

\

Obviously, $Check$ makes $O(\log k)$ queries. The correctness of $Check$ can be seen as follows: 
Observe that (i) if $L'$ and $R'$ are of different weight, at least one of Steps 2.1.2 and 2.1.3 is tilted, 
and (ii) if $L'$ and $R'$ are of the same weight, then both of Steps 2.1.2 and 2.1.3 are balanced 
if and only if $L$ and $R$ are of the same weight. Hence the algorithm essentially verifies 
if $Y_1$ and $Y_2$ are of the same weight, $Y_1\cup Y_2$ and $Y_3\cup Y_4$ 
are of the same weight, $Y_1\cup\cdots\cup Y_4$ and $Y_5\cup\cdots \cup Y_8$ 
are of the same weight, and so on. If all the tests go through, then~$Y_1$ through~$Y_{k+1}$ are all the same weight, which cannot happen if $\overline{X}_1$ includes false coins since $\overline{X}_1$ includes at most $k$ such ones. 

Finally, we adapt our algorithm so that it can satisfy the big-pan property.
We simulate the transformation $|\widetilde{q}\rangle\mapsto(-1)^{\widetilde{q}\cdot x}|\widetilde{q}\rangle$
of the IP oracle by replacing a query string $\widetilde{q}\in\{0,1\}^N$ 
with even Hamming weight $l$ by two queries with Hamming weight $\lfloor N/2\rfloor$ 
when $l/2$ is even (similarly for the case where it is odd). 
We replace $\widetilde{q}$ by two $N$-bit strings $\widetilde{q}_1$ and $\widetilde{q}_2$
with Hamming weight $l/2$ such that $\widetilde{q}=\widetilde{q}_1\oplus \widetilde{q}_2$. 
We take an arbitrary $N$-bit string $\widetilde{b}$ with 
$wt(\widetilde{b})=\lfloor N/2\rfloor-l/2$ such that $I(\widetilde{b})\cap I(\widetilde{q})=\emptyset$.
Note that both $\widetilde{q}_1\oplus \widetilde{b}$ and
$\widetilde{q}_2\oplus\widetilde{b}$ have Hamming weight~$\lfloor N/2\rfloor$. 
(Recall that the Hamming weight of query strings must be even. 
So, if $\lfloor N/2\rfloor$ is odd, then we need an adjustment ($-1$) 
of the Hamming weight when selecting $\widetilde{b}$.)
Since  $(-1)^{(\widetilde{q}_1\oplus \widetilde{b})\cdot x}
(-1)^{(\widetilde{q}_2\oplus \widetilde{b})\cdot x}=(-1)^{\widetilde{q}\cdot x}$
for any $x$, we can replace a query $\widetilde{q}$ to the IP oracle
by two queries $\widetilde{q}_1\oplus \widetilde{b}$ and $\widetilde{q}_2\oplus \widetilde{b}$.
Thus, we can simulate $Find^*(k)$ without changing the complexity (up to a constant factor).

\section{Other Lower Bounds for Restricted Pans}\label{sec:pan_size}

In addition to Theorem \ref{bigpan}, we can show more lower bounds 
for the case where the size of pans is restricted. 
In what follows, $L\leq l$ denotes the restriction that at most $l$ coins 
must be placed on the pans whenever the balance is used. 

First, we give a lower bound for the case where the size of pans is ``small.''
Note that Theorem \ref{smallpan} implies that there is no $o(k^{1/4})$-query algorithm 
placing at most $O(N/\sqrt{k})$ coins on the pans whenever the balance is used.  
 
\begin{theorem}\label{smallpan}
If $L\leq l$, then we need $\Omega(\sqrt{kN/l\min(k,l)})$ weighings. 
In particular, we need $\Omega(\sqrt{N/l})$ weighings. 
\end{theorem}

\begin{proof}
For simplicity, the following weight scheme %, which is similar to that of Theorem \ref{quasi_lower}, 
is given when the size of each pan is~$l$ (that is, when $q\in Q_{l}$). 
But the same bound is also obtained similarly when the size is at most $l$, 
and hence we can apply Lemma \ref{adv_boracle} for $Q=\bigcup_{l'\leq l}Q_{l'}$ to obtain the desired bound 
in the last of this proof. Let $S=\{x\in\{0,1\}^N \mid wt(x)=k\}$ and $f(x)=x$. 
For $(x,y)\in S\times S$, let $w(x,y)=1$ if ${d}(x,y)=2$ (where ${d}(x,y)$ denotes 
the Hamming distance between $x$ and $y$) and $0$ otherwise. 
When the query $q$ for $x$ means that $m_1$ and $m_2$ false coins 
are placed on the left and right pans, respectively, 
and $q$ for $y$ means that $m_3$ and $m_4$ false coins 
are placed on the left and right pans, respectively, 
we put the same weight for all $w'(x,y,q)$'s of such triples 
$(x,y,q)$, which is denoted as $w'((m_1,m_2),(m_3,m_4))$.
Then we define
\begin{align*}
& w'((m_1,m_2),(m_3,m_4))\\ 
&=\left\{
\begin{array}{ll}
\frac{m(N-k-(2l-2m))}{(l-m+1)(k-2m+1)} &\mbox{if}\ (m_1,m_2,m_3,m_4)=(m-1,m,m,m),(m,m-1,m,m),\\
\frac{(l-m+1)(k-2m+1)}{m(N-k-(2l-2m))} &\mbox{if}\ (m_1,m_2,m_3,m_4)=(m,m,m-1,m),(m,m,m,m-1),\\
1 &\!\!\!\!\!\!\!\!\!\!\!\!\!\!\!\!\!\!\!\!\!\!\!\! \mbox{if one of $m_i$'s is $m$ and the others are $m-1$, or}\\ %(m_1,m_2,m_3,m_4)=(m,m-1,m-1,m-1),(m-1,m,m-1,m-1),(m-1,m-1,m,m-1),(m-1,m-1,m-1,m),\\
  &\!\!\!\!\!\!\!\!\!\!\!\!\!\!\!\!\!\!\!\!\!\!\!\! (m_1,m_2,m_3,m_4)=(m+1,m-1,m,m),(m-1,m+1,m,m),\\
  &\!\!\!\!\!\!\!\!\!\!\!\!\!\!\!\!\!\!\!\!\!\!\!\! (m,m,m+1,m-1),(m,m,m-1,m+1),\\
0 &\!\!\!\!\!\!\!\!\!\!\!\!\!\!\!\!\!\!\!\!\!\!\!\! \mbox{otherwise},  
\end{array}
\right.
\end{align*}
where $1\leq m\leq \min(k/2,l)$. It is easy to see that the condition of a weight scheme is satisfied.
Notice that for any $x\in S$ we have $\mu(x)=k(N-k)$. 
Evaluating $\nu(x,q)$ is a bit complicated.
Since this value depends on the numbers of false coins on 
the two pans, $m_1$ and $m_2$, represented by the pair $(x,q)$,
we denote it by $\nu(m_1,m_2)$. We want to evaluate $\nu(x,q)\nu(y,q)$ 
such that $w(x,y)>0$, i.e., ${d}(x,y)=2$ and $\chi(x;q)\neq\chi(y;q)$.
By symmetry, we can assume that $\chi(x;q)=1$ and $\chi(y;q)=0$.
Since ${d}(x,y)=2$, we need to consider only the following cases:
(i) $\nu(x,q)=\nu(m,m-1)$ (or $=\nu(m-1,m)$) and $\nu(y,q)=\nu(m,m)$ (where $0< m \leq \min(k/2,l)$); 
(ii) $\nu(x,q)=\nu(m+1,m-1)$ (or $=\nu(m-1,m+1)$) and $\nu(y,q)=\nu(m,m)$ (where $0< m <\min(k/2,l)$); 
(iii) $\nu(x,q)=\nu(m+1,m)$ (or $=\nu(m,m+1)$) and $\nu(y,q)=\nu(m,m)$ (where $0\leq m <\min(k/2,l)$).
In case of~(i), 
\begin{align*}
\nu(x,q) &= \sum_{y:{d}(x,y)=2,\ \chi(y;q)=0} w'(x,y,q)\\
&= w'((m,m-1),(m-1,m-1))\times m(N-k-(2l-(2m-1)))\\
& \ \ +w'((m,m-1),(m,m))\times (l-(m-1))(k-(2m-1))\\
&\leq 2m(N-k-2l+2m)\\
&= O(\min(k,l)N),
\end{align*}
and
\begin{align*}
\nu(y,q) &= \sum_{x:{d}(x,y)=2,\ \chi(x;q)=1} w'(y,x,q)\\
&= (w'((m,m),(m+1,m))+w'((m,m),(m,m+1))) (l-m)(k-2m)\\
&+ (w'((m,m),(m,m-1))+w'((m,m),(m-1,m))) m(N-k-(2l-2m))\\
&+ (w'((m,m),(m+1,m-1))+w'((m,m),(m-1,m+1))) m(l-m)\\
&= 2(l-m)(k-m)+2(l-m+1)(k-2m+1)\\
&= O(kl),
\end{align*}
and hence $\nu(x,q)\nu(y,q)=O(kNl\min(k,l))$. Similarly, in case of (iii), it holds 
that $\nu(x,q)\nu(y,q)=O(kNl\min(k,l))$.
In case of (ii), 
\begin{align*}
\nu(x,q) &= w'((m+1,m-1),(m,m))\times(l-(m-1))(m+1)\\
&= (l-m+1)(m+1)=O(\min(k/2,l)l)=O(\min(k,l)N),
\end{align*}
and $\nu(y,q)=O(kl)$, and hence we also have $\nu(x,q)\nu(y,q)=O(kNl\min(k,l))$.
From the above, by Lemma \ref{adv_boracle} the quantum query complexity is at least
\[
\Omega\left(
\min_{\substack{x,y,q:\ w(x,y)>0,\\ \chi(x;q)\neq \chi(y;q) }}
\sqrt{ \frac{\mu(x)\mu(y)}{\nu(x,q)\nu(y,q)} }
\right)
%=\Omega\left(\sqrt{\frac{kN\times kN}{kNl\min(k,l)}}\right)
=\Omega\left(\sqrt{\frac{kN}{l\min(k,l)}}\right).
\]
This completes the proof. %\hfill$\Box$
\end{proof}

Second, we generalize Theorem \ref{bigpan} to the case where 
``big pans'' and ``small pans'' are both available but ``medium pans'' 
are not. Here, ``$L\leq l_1$ or $L\geq l_2$'' means that 
at most $l_1$ coins or at least $l_2$ coins (or their superposition) 
must be placed on the pans whenever the balance is used. 

\begin{theorem}\label{bigpan_generalized}\sloppy
If $L\leq l_1$ or $L\geq l_2$ where $l_1<l_2$, 
then we need $\Omega(\min((N/l_1k)^{1/2},(l_2k/N)^{1/4}))$ weighings.
In particular, for any $\epsilon\geq 0$, if $L\leq N/k^{1+2\epsilon}$ 
or $L\geq N/k^{1-4\epsilon}$, then we need $\Omega(k^{\epsilon})$ weighings.
\end{theorem}

\begin{proof}
We can use the same weight scheme as the proof of Theorem \ref{bigpan}. 
Let $l_1=N/d_1$ and $l_2=N/d_2$ with $d_1>d_2$. 
The lower bound we should show is $\Omega(\min((d_1/k)^{1/2},(k/d_2)^{1/4}))$. 
Similar to the proof of Theorem \ref{bigpan}, we can show that by Lemma \ref{adv_boracle} 
the quantum query complexity is at least
\begin{equation}\label{eq091024-2}
\Omega\left(
\min_{\substack{ x,y,q\\ w(x,y)>0\\ \chi(x;q)\neq \chi(y;q) }}
\!\!\!\!\!\!\!\sqrt{ \frac{\mu(x)\mu(y)}{\nu(x,q)\nu(y,q)} }
\right)
=\Omega\left(
\min_{\substack{ c\\ c\geq d_1\\ \mbox{{\scriptsize or}}\ \leq d_2}}
\!\!\!\sqrt{\frac{\binom{N}{k}}{\gamma(N,k,c)}\cdot \frac{\binom{N}{k}}{\binom{N}{k}-\gamma(N,k,c)}}
\right).
\end{equation}
Then the theorem can be obtained from Eq.(\ref{eq091024-2}) 
by using Lemma \ref{balance_sum} for $c\leq d_2$ and the following lemma 
(Lemma \ref{balance_sum2}) for $c\geq d_1$. (Notice that it suffices to show Lemma \ref{balance_sum2} for $c\geq 3$ 
since the bound we should obtain from Lemma \ref{balance_sum2}, 
$(d_1/k)^{1/2}$, is nontrivial only if $d_1=\omega(k)$ 
and hence the size of pans $N/c\ (\leq l_1)$ should be considered only for $c=\omega(k)$).

\begin{lemma}\label{balance_sum2}
$(\binom{N}{k}-\gamma(N,k,c))/\binom{N}{k} =O(\frac{k}{c})$ 
for any $c\geq 3$. %satisfying $k^{3/2} \le c \le N/2$.
\end{lemma}

\begin{proof}
Let us bound the probability that the scale is tilted when $N/c$ coins 
($N$ coins include $k$ false ones) are randomly placed on each of the two pans 
since it is exactly $(\binom{N}{k}-\gamma(N,k,c))/\binom{N}{k}$. 
%Note that the scale is tilted when there are odd number of false coins placed on the two pans. 
Clearly, this probability is upper bounded by the sum $\sum_{m=1}^k t'(m)$ 
where $t'(m):=\frac{\binom{k}{m}\binom{N-k}{2N/c-m}}{\binom{N}{2N/c}}$ denotes 
the probability of choosing exactly $m$ false coins out of $k$ ones 
when $2N/c$ coins are placed on the pans. 
Letting $r'(m) := t'(m+1)/t'(m) = \frac{(k-m)(2N/c-m)}{(m+1)(N-k-2N/c+m+1)}$, the sum is bounded by
\begin{align*}
\sum_{m=1}^k t'(m) 
&\leq (r'(0)+r'(0)^2+\cdots)t'(0)\ \ \ \mbox{(since $r'(0) \ge r'(m)$ for all $m \ge 1$)}\\
&\leq \frac{r'(0)}{1-r'(0)}\ \ \ \mbox{(by $t'(0)\leq 1$)} \\ 
&= O(r'(0)).
\end{align*}
Since $c\geq 3$ and $k\leq N/2$, we can see that the following holds: %for $ k^{3/2} \le c \le N/2$.
$$
r'(0) \le \frac{2kN}{cN-ck-2N} = \frac{2k}{c} \cdot 
 \frac{1}{1-\frac{k}{N}-\frac{2}{c}}=  O(k/c).
$$
This completes the proof. %\hfill$\Box$
\end{proof}
Hence the proof of Theorem \ref{bigpan_generalized} is completed. %\hfill$\Box$
\end{proof}

Unfortunately, Theorem \ref{bigpan_generalized} does not give even a weakest 
nontrivial lower bound $\omega(1)$ if the size of the pans is not restricted. 
One might have the hope by Theorem~\ref{bigpan_generalized} 
that we could obtain a good upper bound by always placing approximately $N/k$ coins on the pans,
but Theorem \ref{smallpan} denies such a hope since 
we have an $\Omega(k^{1/2-2\epsilon})$ lower bound for $l=N/k^{1-4\epsilon}$.

\end{document}